\documentclass[twoside]{article}
\usepackage{qic,epsfig, epstopdf}

\usepackage{amsmath}

\usepackage{amssymb}

\usepackage{amsthm}
\usepackage{enumerate}
\usepackage{bbold}
\usepackage{dsfont}
\usepackage{graphicx}
\usepackage{color}
\usepackage[numbers]{natbib}
\usepackage{url}
\usepackage[colorlinks=true,urlcolor=webblue,linkcolor=webgreen,filecolor=webblue,citecolor=webgreen,pdfpagemode=UseOutlines,pdfstartview=FitH,pdfpagelayout=OneColumn,bookmarks=true]{hyperref}
\usepackage{epstopdf}

\input{Qcircuit}
\newtheorem{thm}{Theorem}[section]

\newtheorem{de}{Definition}
\theoremstyle{definition}
\theoremstyle{remark}
\newtheorem{bb}{\emph{\textbf{Problem}}}

\definecolor{webgreen}{rgb}{0,.5,0}
\definecolor{webblue}{rgb}{0,0,.5}

\usepackage{thmtools}
\declaretheoremstyle[notefont=\bfseries,notebraces={}{},%
    headpunct={},postheadspace=1em]{mystyle}
\declaretheorem[style=mystyle,numbered=no,name=Theorem]{thm-hand}

\def\dist{{\textbf{dist}}}
\def\TV{{\text{TV}}}
\def\MCG{{\text{MCG}}}
\def\WRT{{\text{WRT}}}

\def\H{\mathcal{H}}
\def\D{\mathcal{D}}

\def\RR{\mathbb{R}}
\def\CC{\mathbb{C}}
\def\P{\text{P}}
\def\NP{\text{NP}}
\def\FP{\text{FP}}
\DeclareGraphicsExtensions{.jpg,.pdf}

\newcommand{\innerprod}[2]{\langle #1| #2\rangle}
\newcommand{\opn}[1]{\operatorname{#1}}
\newcommand{\one}{\mathds 1}
\newcommand{\expref}[2]{\texorpdfstring{\hyperref[#2]{#1~\ref{#2}}}{#1~\ref{#2}}}

\begin{document}
\setlength{\textheight}{8.0truein}    

\runninghead{Quantum invariants of 3-manifolds and NP vs \#P}{Gorjan Alagic and Catharine Lo}

\normalsize\textlineskip
\thispagestyle{empty}
\setcounter{page}{125}

\copyrightheading{17}{1\&2}{2017}{0125--0146}

\vspace*{0.88truein}

\alphfootnote

\fpage{125}


\centerline{\bf QUANTUM INVARIANTS OF 3-MANIFOLDS AND NP VS \#P}
\vspace*{0.37truein}
\centerline{\footnotesize
GORJAN ALAGIC\footnote{galagic@gmail.com}}
\vspace*{0.015truein}
\centerline{\footnotesize\it QMATH, Department of Mathematical Sciences, University of Copenhagen}
\baselineskip=10pt
\centerline{\footnotesize\it Copenhagen, Denmark}
\vspace*{10pt}
\centerline{\footnotesize 
CATHARINE LO\footnote{cwlo@math.princeton.edu}}
\vspace*{0.015truein}
\centerline{\footnotesize\it Department of Mathematics, Princeton University}
\baselineskip=10pt
\centerline{\footnotesize\it Princeton, New Jersey, USA}
\vspace*{0.225truein}
\publisher{September 25,2016}{January 2, 2017}

\vspace*{0.21truein}

\abstracts{The computational complexity class $\#\P$ captures the difficulty of counting the satisfying assignments to a boolean formula. In this work, we use basic tools from quantum computation to give a proof that the $SO(3)$ Witten-Reshetikhin-Turaev (WRT) invariant of 3-manifolds is $\#\P$-hard to calculate. We then apply this result to a question about the combinatorics of Heegaard splittings, motivated by analogous work on link diagrams by M. Freedman. We show that, if $\#\P \not\subseteq \FP^\NP$, then there exist infinitely many Heegaard splittings which cannot be made logarithmically thin by local WRT-preserving moves, except perhaps via a superpolynomial number of steps. We also outline two extensions of the above results. First, adapting a result of Kuperberg, we show that any presentation-independent approximation of WRT is also $\#\P$-hard. Second, we sketch out how all of our results can be translated to the setting of triangulations and Turaev-Viro invariants.}{}{}

\vspace*{10pt}

\keywords{}
\vspace*{3pt}
\communicate{R~Cleve~\&~R~de~Wolf}

\vspace*{1pt}\textlineskip
\section {Introduction}

\subsection{Background}

In computational complexity, a significant amount of useful theory rests on a foundation of widely-believed conjectures about the separation of complexity classes. The most famous example is the conjecture $\P \neq \NP$, which states that polynomial-time computability is a much stronger condition than polynomial-time verifiability. Besides the well-known P and NP, another class of interest is $\#$P, which is characterized by the problem of counting the number of satisfying assignments to a boolean formula. It is straightforward to see that $\P \subseteq \NP \subseteq \#\P$, which amounts to saying that verifying a solution to a satisfiability problem is easier than finding one, which in turn is easier than finding all of them\footnote{Strictly speaking, $\P$ and $\NP$ are decision classes, and as such are not directly comparable to $\#\P$, which is a counting class. So formally we are talking here about the inclusions $\P \subseteq \NP \subseteq \P^{\#\P}$}. Both containments are conjectured to be strict, and complexity theorists have amassed a significant amount of evidence for the truth of these conjectures. This is based on a large amount of theoretical work, as well as decades of practical experience; for instance, one consequence of $\P = \NP$ would be that practical cryptography does not exist.

In low-dimensional topology, complexity theory has an important role to play, e.g., in characterizing the computational difficulty of classification. While surfaces can be distinguished in polynomial time via the Euler characteristic, the classification problem appears to be harder for knots and 3-manifolds. The problem of deciding if a given knot diagram represents the unknot is known to be in both NP and coNP~\cite{Hass, coNP}, but no polynomial-time algorithm is known. Similar results hold for the three-sphere~\cite{HassKuperberg}. The difficulty of calculating certain invariants turns out to be greater still. For instance, it is $\#$P-hard to exactly calculate the Jones Polynomial~\cite{Vertigan1}. While approximate calculation of the Jones Polynomial is possible in quantum polynomial time~\cite{AJL, FKLW}, these approximations are likely too weak to be useful~\cite{Kuperberg}. For the case of quantum invariants of 3-manifolds, some analogous results are known. Using explicit formulas, Kirby and Melvin showed that the $SU_r(2)$ Witten-Reshetikhin-Turaev (or WRT) invariant is in $\P$ for $r = 4$, but $\#\P$-hard\footnote{Although they only claimed $\NP$-hardness, an inspection of their proof shows that, in fact, they showed $\#\P$-hardness. This was recently extended to the Turaev-Viro case, also at $r = 6$~\cite{Burton}.}\,\, for $r = 6$~\cite{KirbyMelvinArf}. All of the WRT invariants, as well as the closely-related Turaev-Viro (or TV) invariants,  can be approximated in quantum polynomial time; at infinitely many levels, they are also quantum-universal~\cite{BQP3, BQP2, FLW, BQP1}.

Given the above, it is natural to ask: do the strongly-believed conjectures of complexity theory have consequences in low-dimensional topology? An intriguing idea of M. Freedman combines the conjectured separation $\NP \neq \#\P$ with the $\#\P$-hardness of the Jones polynomial to prove the existence of link diagrams with curious properties~\cite{Freedman1}. These diagrams cannot be made ``logarithmically thin'' by local transformations, unless one is permitted superpolynomially-many moves. Crucially, one is allowed to apply both Reidemeister moves and a certain geometry-breaking move called the ``$r$-move''. Freedman also argues that, without the complexity conjecture, this fact appears to be resistant to known techniques in low-dimensional topology. If one also assumes certain separations between quantum complexity classes, even stronger results of this kind are achievable~\cite{Freedman2}.

\subsection{The present work}

Our results are primarily concerned with the $SO(3)$ WRT and TV invariants at $r$th roots of unity, where $r$ is prime and at least five. These are the cases for which the crucial density result of Larsen and Wang holds~\cite{LarsenWang}. Our first result shows that these invariants are $\#$P-hard to calculate exactly, or even to approximate with a certain additive error. We will state our results for the case of WRT, and put off extensions to TV to later sections. We assume that the input manifold is specified via a Heegaard splitting, i.e. a pair $(g, \alpha)$ where $g > 0$ is an integer and $\alpha$ is a word in the standard generators of the mapping class group MCG($\Sigma_g$) of the genus-$g$ surface $\Sigma_g$. We denote the class of 3-manifolds homeomorphic to $(g, \alpha)$ by $M_{g, \alpha}$. Given a positive integer $r$, we denote the WRT invariant  attached to $SO(3)$ at the $r$-th root of unity by $\WRT_r$.

\begin{thm}\label{thm:intro1}
Let $r \geq 5$ be prime. The following problem is $\#\P$-hard: given an integer $g$ and a word $\alpha$ in the canonical generators of $\MCG(\Sigma_g)$, output $\WRT_r(M_{g, \alpha}).$
\end{thm}

In fact, by adapting a theorem of Kuperberg~\cite{Kuperberg} about the Jones polynomial, we can prove an even stronger result, showing that any value-distinguishing approximation of the invariant is hard. 

\begin{thm}\label{thm:intro3}
Let $r \geq 5$ be prime, and $0 < a < b$ real. Given an integer $g$, a word $\alpha$, and a promise that $|\WRT_r(M_{g, \alpha})| < a$ or $|\WRT_r(M_{g, \alpha})| > b$, it is $\#\P$-hard to decide which is the case.
\end{thm}

Our second result is an analogue of Freedman's aforementioned work on complexity-theoretic conjectures and link diagrams~\cite{Freedman1}. To state the result, we need to set down a notion of distance between Heegaard splittings. This is defined in terms of basic moves: stabilizations, handle-slides, and a so-called ``$r$-move''. One $r$-move consists of choosing a Dehn twist $\sigma$ among the $3g-1$ canonical generators of MCG($\Sigma_g$), and inserting $\sigma^{4r}$ anywhere in the word. For any particular (positive, integral) $r$, this gives a notion of distance (called $r$-distance) between a pair of Heegaard splittings: we simply take the length of the shortest sequence of stabilizations, handle-slides, and $r$-moves which transforms one splitting into the other. We denote the $r$-distance by $\dist_r(\cdot, \cdot)$ and we say that two splittings are $r$-related (written $\sim_r$) if the $r$-distance between them is finite. It is then straightforward to show that
$$
(g, \alpha) \sim_r (h, \beta)
\qquad \text{implies} \qquad
\WRT_r(M_{g, \alpha}) = \WRT_r(M_{h, \beta})\,.
$$
An analogous fact was already known for the absolute value of WRT for surgery presentations~\cite{Gilmer, Lackenby, Wong}. Conditioned on the conjecture $\#\P \not\subseteq \FP^\NP$, we can prove the existence of a family of Heegaard splittings which cannot be made logarithmically thin (via $\sim_r$) except perhaps in superpolynomial time.

\begin{thm}\label{thm:intro2}
Assume $\#\P \not\subseteq \FP^{\NP}$, and choose prime $r \geq 5$. For any polynomial $p$, there exists an infinite family of Heegaard splittings $(g_j, \alpha_j)$ with the following property: for any family $(h_j, \alpha_j)$ satisfying $\dist_r((g_j, \alpha_j), (h_j, \beta_j)) \leq O(p(g_j))$, it is the case that $h_j \in \Omega(\log(g_j))$.
\end{thm}

The class $\FP$ is simply the function version of $\P$, i.e., one is allowed to output a bitstring rather than simply ``accept'' or ``reject''. We remark that the statement $\#\P \not\subseteq \FP^{\NP}$ is slightly weaker than $\#\P \neq \NP$; in fact, the statement $\#\P \neq \NP$ is technically a type violation since the former is a counting class while the latter is a decision class.

We remark that, while our exposition focuses on the $SO(3)$ version of the invariant, our results hold for any category where the relevant representation of the mapping class group is dense in the projective unitary group. In particular, our results also hold for the $SU(2)$ theory for $r \equiv 1 \bmod 4$ (see concluding remarks in ~\cite{LarsenWang}.)

\subsection{Organization}

The paper is organized as follows. In the next section, we recall some background facts regarding WRT invariants of 3-manifolds. We explain how to explicitly compute these invariants for Heegaard splittings, and how to define the $r$-move and the $r$-distance. We also state the density theorem of Larsen and Wang~\cite{LarsenWang}, which is crucial to our results. In \expref{Section}{sec:quantum}, we briefly discuss some basic facts from complexity theory and prove a fact about the $\#\P$-hardness of calculating matrix entries of quantum circuits\footnote{This fact is weaker than Aaronson's well-known theorem stating that PostBQP = PP~\cite{AaronsonPostBQP}, but it is easier to state and sufficient for most of what we do.}. \expref{Section}{sec:results} contains the proofs of \expref{Theorem}{thm:intro1} and \expref{Theorem}{thm:intro2}. In \expref{Section}{sec:tv}, we outline two extensions. First, we show how to adapt a theorem of Kuperberg (Theorem 1.2 in~\cite{Kuperberg}) in order to achieve the stronger $\#\P$-hardness result, i.e., \expref{Theorem}{thm:intro3}. Second, we provide a sketch of how all of our results can be extended to the setting of Turaev-Viro invariants and triangulated 3-manifolds. The exposition in Sections 1 to 4 is largely self-contained, and should be relatively accessible. On the other hand, \expref{Section}{sec:tv} requires a number of technical results in complexity theory (for which we refer to ~\cite{Kuperberg}) as well as some background on Turaev-Viro invariants, tensor networks, and triangulations (for which we refer to ~\cite{BQP3, TV}.)

\section{3-manifolds and quantum invariants}

\subsection{3-manifolds, Heegaard splittings, and homeomorphism-preserving moves}\label{sec:manifolds}

In this work, a ``3-manifold'' will mean a compact connected Hausdorff space, each point of which has a neighborhood homeomorphic to $\RR^3$. It is well-known that any 3-manifold $M$ can be described by a positive integer $g$ and a self-homeomorphism $\alpha$ of $\Sigma_g$, the compact orientable surface of genus $g$. Specifically, $M$ is homeomorphic to the quotient space $H_g \sqcup_\alpha H_g$ formed by gluing two genus-$g$ handlebodies along $\alpha: \partial H_g \rightarrow \partial H_g$. The group $\MCG(\Sigma_g)$ of (isotopy classes of) self-homeomorphisms of $\Sigma_g$ is generated by Dehn twists $\{\sigma_1, \sigma_2, \dots, \sigma_{3g-1}\}$ about the canonical curves shown in \expref{Figure}{canonical-curves}. This leads to a purely combinatorial description of any 3-manifold, as a pair $(g, \alpha)$ where $g$ is a positive integer and $\alpha$ is a word in the $3g-1$ generators; moreover, any such pair describes a valid 3-manifold. We will refer to such a pair as a Heegaard splitting, and write $M_{g, \alpha}$ for the corresponding homeomorphism class of 3-manifolds. We will frequently discuss the complexity of computational tasks whose input is a Heegaard splitting $(g, \alpha)$. It is then implicit that the number of computation steps is measured as a function of the input size, which we take to be $|\alpha|$, i.e., the length of the word $\alpha$.

\begin{figure}[h]
\begin{center}
\includegraphics{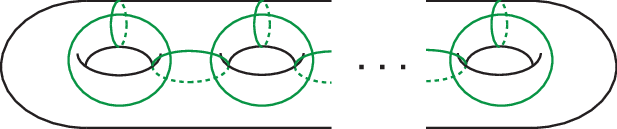}
\fcaption{\label{canonical-curves} The set of Dehn twists about the $3g-1$ canonical curves generates MCG$(\Sigma_g)$.}
\end{center}
\end{figure}

If two Heegaard splittings represent homeomorphic 3-manifolds, then there exists a finite sequence of homeomorphism-preserving ``moves'' which transforms one splitting into the other~\cite{Singer}. There are two types of moves. The first is called the ``handle-slide'', and it maps 
$$
(g, \alpha) \mapsto (g, \beta \alpha \beta')
$$
where $\beta$ and $\beta'$ are words which describe two elements of the so-called handlebody subgroup. The handlebody subgroup consists of self-homeomorphisms of $\Sigma_g$ which extend to self-homeomorphisms of the handlebody $H_g$. The second move type is ``stabilization'', and is defined by
\begin{equation}\label{eq:stabilize}
(g, \alpha) \mapsto (g+1, \alpha \sigma_{3g+1} \sigma_{3g+2} \sigma_{3g+1}^{-1})\,.
\end{equation}
In the right-hand side above, we have implicitly interpreted $\alpha$ as a word in the generators of $\MCG(\Sigma_{g+1})$ in the obvious way. Stabilization amounts to taking the connected sum of $M_{g, \alpha}$ with the three-sphere, described by the genus-one Heegaard splitting $(1, \sigma_1\sigma_2\sigma_1^{-1})$. The inverse move (destabilization) can only be applied if we can ``undo'' a connected sum with the three-sphere, i.e.,
\begin{equation}\label{eq:destabilize}
(g, \gamma \sigma_{3g-2} \sigma_{3g-1} \sigma_{3g-2}^{-1}) \mapsto (g-1, \gamma)\,,
\end{equation}
where the word $\gamma$ may not contain the generators $\sigma_{3g-1}$ and $\sigma_{3g-2}$.\footnote{It is not necessary to consider other genus-one splittings of the sphere: they are all related only by handle-slides~\cite[Theorem 3.7]{Schar}. It is also not necessary to consider words $\alpha'$ which satisfy the destabilization conditions and are equivalent to $\alpha := \gamma \sigma_{3g} \sigma_{3g-1} \sigma_{3g}^{-1}$; since $1$ is obviously in the handlebody subgroup, the transformation $\alpha \mapsto \alpha'$ is a valid handle-slide. It thus suffices to consider the particular choice of stabilization and destabilization move we chose.} 

The above ideas allow us to define a notion of distance between Heegaard splittings. First, we define a length $|\cdot|$ for moves. The length of a handle-slide $(g, \alpha) \mapsto (g, \beta\alpha\beta')$ is the sum $|\beta| + |\beta'|$ of the lengths of the relevant words. A stabilization move is described by a single bit, which determines if the genus should increase or decrease. In either case, the length of a stabilization move is one. The length $|s|$ of a sequence $s$ of handle-slides and stabilization moves is simply the sum of the lengths of all the moves in the sequence. Now we are ready to define a distance between a pair $(g, \alpha)$ and $(g', \alpha')$ of Heegaard splittings:
$$
\dist((g, \alpha), (g', \alpha')) := 
\min \{ |s| : s \text{ is a sequence of moves with } s(g, \alpha) = (g', \alpha') \}\,.
$$
If no sequence of moves suffices, then the distance is defined to be infinite. If the distance is finite, we will write $(g, \alpha) \sim (g', \alpha')$ and say that the two Heegaard splittings are equivalent. This precisely captures the notion of 3-manifold homeomorphism.
\begin{thm}\label{thm:heegaard}
~\cite{Rei, Singer} Let $(g, \alpha)$ and $(g', \alpha')$ be two Heegaard splittings. Then $\dist((g, \alpha), (g', \alpha'))$ is finite if and only if the corresponding 3-manifolds $M_{g, \alpha}$ and $M_{g', \alpha'}$ are homeomorphic.
\end{thm}

\subsection{The WRT invariant of 3-manifolds, and $r$-distance}\label{sec:wrt}

We now sketch the definition of the $SO(3)$ WRT invariants of 3-manifolds. A thorough description is given in~\cite{WRT}. Our definition will be based on the Heegaard splitting, rather than Dehn surgery. This version of the invariant is sometimes referred to as the Crane-Kohno-Kontsevich presentation; it was proved equivalent to the surgery formulation by Piunikhin~\cite{Piunikhin93}. 

The WRT invariants form an infinite family, parameterized by a prime $r \geq 5$. Fixing such an $r$, we set $A = ie^{2\pi i/ 4r}$ and define the quantum integers
$$
[k] = \frac{A^{2k} - A^{-2k}}{A^2 - A^{-2}}\,.
$$
Each choice of $r$ also comes with the following data:
\begin{enumerate}
\item a finite set $L$ of integer ``labels'' $\{0, 2, 4, \dots, r-3\}$;
\item a list $d : L \rightarrow \CC$ of ``dimensions'' for each label, given by $d_i = [i+1]$; 
\item a finite set $O \subset L^3$ of ``fusion rules'';
\item an ``$R$-tensor'' $R_i^{jk} : L^3 \rightarrow \CC$;
\item an ``$F$-tensor'' $F_{ijk}^{lmn} : L^6 \rightarrow \CC$.
\end{enumerate}

\noindent There are various constraints on the above data. For instance, $(i, j, k) \in O$ implies that there is a planar (but possibly degenerate) triangle with sidelengths $i$, $j$ and $k$; the $R$-tensor is zero unless its indices are set to a triple from $O$; the $F$-tensor is zero unless the sextuple of indices is the set of sidelengths of a rigid tetrahedron. If we fix two lower and two upper indices (say $i,j,l,m$,) the $F$-tensor becomes a matrix (i.e., a linear map) $F_k^n$ on $\CC^L$; this map is always unitary (see also \eqref{eq:F-move} below.) For the complete set of constraints, see~\cite{WRT, FreedmanKrushkal, Piunikhin93}. In general, the five pieces of data above completely determine two additional quantities which we will require. These are the ``total dimension'' $\D \in \CC$ and the ``$S$-tensor'' $S : L^3 \rightarrow \CC$ (see~\cite{BQP2} and references therein,) defined by

\begin{equation}\label{eq:D and S-move}
\D = \Bigl(\sum_{j \in L} d_j^2\Bigr)^{1/2}
\qquad \qquad \qquad  
\D S^i_{jk} = \sum_{l\,:\,(j,k,l) \in O} \frac{d_l}{\sqrt{d_i}} F_{ljj}^{ikk} R_l^{kj}R_{l}^{jk}\,.
\end{equation}

Before defining the invariant attached to $r$, we must first define a (genus-indexed) family of representations $\rho_{r, g}$ of $\MCG(\Sigma_g)$. The invariant will later be defined as a particular matrix entry of these representations. Decompose $\Sigma_g$ into three-punctured spheres or ``pants'' as shown in  \expref{Figure}{spines}; we will refer to this as the standard pants decomposition of $\Sigma_g$. Dual to such a pants decomposition is a trivalent graph $\Gamma_g$ called the spine. The spine has one vertex for every pants in the decomposition, and one edge for each meeting between two pants. We will assign labels from $L$ to the edges of $\Gamma_g$. Such a labeling is called fusion-consistent if, for every vertex $v$ of $\Gamma$, the triple of labels on the edges incident to $v$ is in $O$. The vector space of the representation $\rho_{r, g}$ is then the formal linear span of fusion-consistent labelings:
\begin{figure}
\begin{center}
\includegraphics{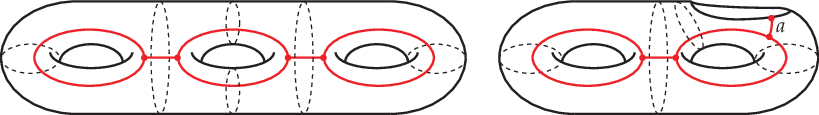}
\fcaption{\label{spines} Pants decompositions (and corresponding spines) of $\Sigma_3$ and $\Sigma_2^{(a)}$.}
\end{center}
\end{figure}
$$
\H_{r,g} := \textbf{span}_\CC \left \{ \ket \ell :  \ell \text{ is a fusion-consistent labeling of } \Gamma_g \right \}. 
$$
This construction naturally extends to surfaces with some finite number of boundary components labeled by elements of $L$. Consider the surface $\Sigma_g^v$ having $g$ handles and $m$ boundary components labeled by the $m$-tuple $v \in L^m$. A pants decomposition of $\Sigma_g^v$ yields a spine $\Gamma_g^v$, which has $m$ edges whose labels are permanently fixed by $v$, as in the example in \expref{Figure}{spines}. We then consider labelings of the remaining edges, so that the total labeling is fusion-consistent. The space spanned by these labelings is denoted $\H_{r,g}^v$.
 
Returning to the case of no boundary, we now define the action of the canonical generators $\sigma_j$ of $\MCG(\Sigma_g)$ on $\H_{r, g}$. First, suppose the canonical curve $\gamma_j$ is isotopic to a cut in the standard pants decomposition. Let $e_j$ denote the (unique) spine edge intersecting $\gamma_j$. Then the corresponding Dehn twist $\sigma_j$ acts diagonally:
\begin{equation}\label{eq:dehn-twist-action}
\rho_{r, g} (\sigma_j) : \ket \ell \mapsto e^{R_0^{\ell(e_j)\ell(e_j)}} \ket \ell\,.
\end{equation}
The action of the other canonical twists is defined similarly, but requires a basis change. Let $n$ denote the number of edges of $\Gamma$. The $F$-tensor and the $R$-tensor can be viewed as linear operators on $\H_{r, g}$, as shown in the expressions below. These operators act in a geometrically local way, in the sense that they are the identity on all but a constant $(\leq 6)$ number of edges. Each expression below shows how to decompose a standard basis element of $\H_{r, g}$ as a linear combination of elements of another orthonormal basis for the same Hilbert space. This other basis corresponds to a different pants decomposition of $\Sigma_g$. The unitary transformations described in \eqref{eq:F-move} are precisely the equivalences between defining $\H_{r, g}$ over these different choices of spine (i.e., pants decomposition).

\begin{equation} \label{eq:F-move}
\raisebox{-1.3cm}{\includegraphics[scale=.5]{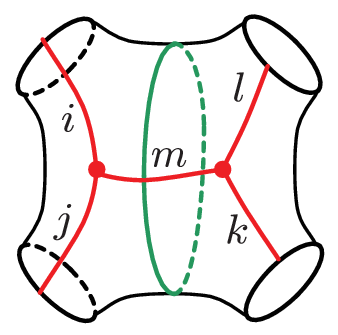}}
=
\sum_n F^{ijm}_{kln} 
\raisebox{-1.3cm}{\includegraphics[scale=.5]{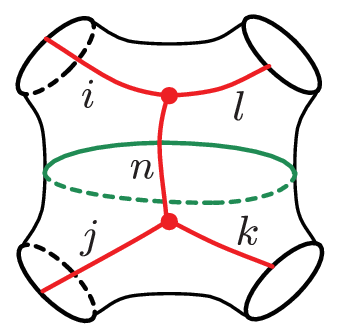}},
\qquad
\raisebox{-1.3cm}{\includegraphics[scale=.5]{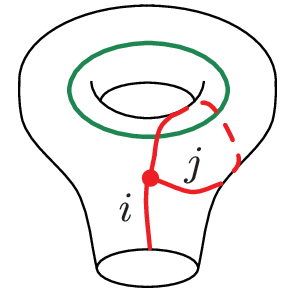}}
=
\sum_k S^i_{jk} 
\raisebox{-1.3cm}{\includegraphics[scale=.5]{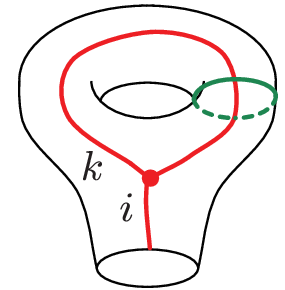}}
\end{equation}

\vspace{0.5cm}
Now let $\gamma_j$ be a canonical curve which is not a cut in the standard pants decomposition. By considering the transformations \eqref{eq:F-move} as well as \expref{Figure}{canonical-curves} and \expref{Figure}{spines}, we see that at most one application each of the $F$-tensor and the $S$-tensor suffices to change the standard pants decomposition into a pants decomposition where $\gamma_j$ is a cut. In that basis (now consisting of labelings of a different spine), $\sigma_j$ acts precisely as in \eqref{eq:dehn-twist-action}. This completely describes the representations $\rho_{r, g}$. Before defining the invariant itself, we record a theorem which is crucial to our results.

\begin{thm}\label{thm:density} \cite{LarsenWang}
For prime $r \geq 5$ and $g > 1$, the image of $\rho_{r, g}$ is dense in $\opn{PSU}(\mathcal H_{r, g})$.
\end{thm}

Finally, we define the WRT invariant as a particular, scaled matrix entry of the representation $\rho_{r, g}$. Let $\ket 0 \in \H_{r, g}$ denote the basis vector corresponding to the labeling where each edge of the spine carries the trivial label $0$.

\begin{de}\label{def:WRT}\cite{Piunikhin93}
Let $(g, \alpha)$ be a Heegaard splitting. The $\text{SO}_r(3)$ Witten-Reshetikhin-Turaev invariant of $(g, \alpha)$ is defined to be
$$
\emph{\WRT}_r(g, \alpha) = \D ^{g-1}\bra 0 \rho_{r, g}(\alpha) \ket 0 \,.
$$
\end{de}
Note that the quantum dimension $\D$ above is implicitly parameterized by $r$. The central fact now is that $\WRT_r(g, \alpha)$ depends only on the homeomorphism type $M_{g, \alpha}$. The proof proceeds by establishing the invariance of $\WRT_r$ under the stabilization and handle-slide moves described above (see \cite{Crane}, \cite{Kohno} and \cite{Kontsevich} for complete proofs). 

In fact, there is a simple additional move which breaks homeomorphism but preserves $\WRT_r$ for any fixed $r$. For any Heegaard splitting $(g, \alpha)$ and any choice of canonical twist $\sigma_j$, the ``$r$-move'' acts by inserting $4r$ twists in any position in the word:
$$
(g, \alpha) \mapsto (g, \alpha_0 \sigma_j^{4r} \alpha_1)
\qquad 
\text{for any words }\alpha_0, \alpha_1 \text{ satisfying } \alpha_0 \alpha_1 = \alpha\,.
$$
Invariance of $\WRT_r$ under the $r$-move follows from the eigenvalues of the Dehn twists~\cite{FreedmanKrushkal}; these are given by $\exp(R_0^{jj}) = \exp(-\pi i j (j+2) / 2r)$, where $j$ is the label of the relevant spine edge. It follows that 
\begin{equation}\label{r-move-invariant}
\rho_{r, g} (\alpha_0 \sigma_j^{4r} \alpha_1) 
= \rho_{r, g} (\alpha_0) \circ \rho_{r, g} (\sigma_j^{4r}) \circ \rho_{r, g} (\alpha_1) 
= \rho_{r, g} (\alpha_0) \circ \one_{\mathcal H_{r, g}} \circ \rho_{r, g} (\alpha_1) 
= \rho_{r, g} (\alpha_0\alpha_1)
= \rho_{r, g} (\alpha)\,.
\end{equation}
In particular, $\WRT_r$ is invariant under the $r$-move. We remark that we can define an inverse $r$-move in the obvious way, i.e., as deletion of any subwords of the form $\sigma_j^{4r}$ for any $j$. 

The $r$-move together with handle-slides and (de)stabilizations defines a second notion of distance on Heegaard splittings. Each $r$-move can be described by the index $j$ of the Dehn twist, and another index indicating where in the word to insert the $4r$ twists. As before, the length of a sequence of moves is defined to be the sum of the lengths of the descriptions of each move. We thus define $r$-distance by setting
$$
\dist_r((g, \alpha), (g', \alpha')) := 
\min \{ |s| : s \text{ is a sequence of moves (incl. }r\text{-moves) with } s(g, \alpha) = (g', \alpha') \}\,.
$$
If no sequence of moves suffices, then the $r$-distance is defined to be infinite. If the $r$-distance is finite, we will write $(g, \alpha) \sim_r (g', \alpha')$ and say that the two Heegaard splittings are $r$-equivalent. We then have the following.
\begin{thm}\label{thm:r-distance}
Let $(g, \alpha)$ and $(g', \alpha')$ be two Heegaard splittings. If $\dist_r((g, \alpha), (g', \alpha'))$ is finite, then $\WRT_r(g, \alpha) = \WRT_r(g', \alpha')$.
\end{thm}

\subsection{A simple example: the Fibonacci representation}

For concreteness, we briefly describe the case $r = 5$, which is the simplest case for which all our results hold. This is the so-called Fibonacci category, and is defined with the following data (see Section 6 of ~\cite{FLW}.) We first compute $A = e^{3 \pi i / 5}$. The data is then
\begin{enumerate}
\item label set $L = \{0, 1\}$ (we relabel $2$ as $1$ for simplicity);
\item dimensions $d_0 = [1] = 1$ and $d_1 = [3] = (1 + \sqrt{5})/2$; 
\item fusion rules $O = \{(a, b, c) \in L^3 : a + b + c \neq 1\}$;
\item $R$-tensor defined by $R_0^{00} = 0$ and $R_0^{11} = 3 \pi i / 5$;
\item $F$-tensor defined by 
\end{enumerate}
\begin{center}
\includegraphics[width = 0.6\textwidth]{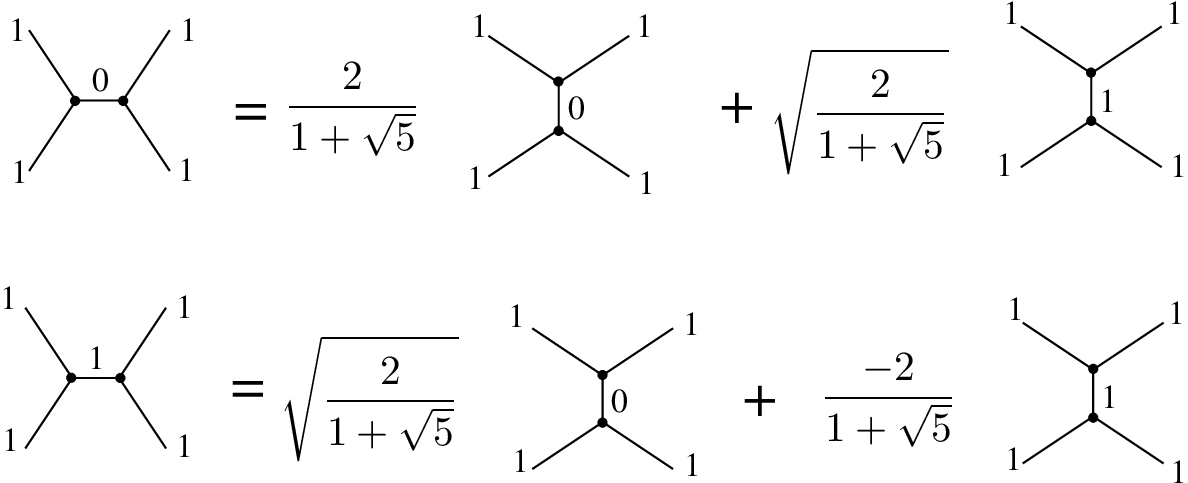}
\end{center}

It's not hard to check that the above values determine the $R$-tensor and $F$-tensor fully: the remaining values are forced by the fusion rules (and are either zero or one). Using equation \eqref{eq:D and S-move}, we see that the total dimension satisfies $\D^2 = (5 + \sqrt{5})/2$ and that the $S$-tensor is defined by
\begin{align*}
\D S_{00}^0 &=1 
\qquad & \qquad
\D S_{10}^0 &= \D S_{01}^0 = d_1\\
\D S_{11}^0 &=1+d_1 e^{i4\pi/5}
\qquad & \qquad
\D S_{11}^1 &=\sqrt{d_1}(1-e^{i4\pi/5})\,.
\end{align*}
As expected from \eqref{eq:dehn-twist-action}, the eigenvalues of the Dehn twists are either $0$ or $e^{3 \pi i / 5}$. It immediately follows that the Fibonacci invariant does not change under insertion of $\sigma_j^{10}$ anywhere in the Heegaard splitting word, for any $j$. Fibonacci was first shown to be dense in~\cite[Theorem 6.2]{FLW}.

\section{Computational complexity and quantum circuits}\label{sec:quantum}

\subsection{P, NP, and $\#$P, briefly}
Computational complexity attempts to classify problems according to the number of basic computation steps required to solve them\footnote{The notion of computation time can be fully formalized using Turing Machines. To gain an intuitive understanding, it is sufficient to think about writing a computer program for the task, and considering the total number of basic instructions (e.g., additions or multiplications) the program must execute on a given input.}, expressed as a function of the input size. We will assume that the computations are all deterministic, and that both the input and the output is a bitstring. An important set of examples concerns the satisfiability problem for boolean formulas. Recall that a boolean formula is an expression involving a finite number of input variables (literals), NOTs (negations), ANDs (conjunctions), and ORs (disjunctions). In this work, we will assume that all formulas are in 3CNF, that is, a conjunction of clauses where each clause is a disjunction of three (possibly negated) literals. Any 3CNF formula can be encoded as a bitstring at linear cost. Note that the number of distinct clauses is at most cubic in the number of variables.

The first relevant problem in satisfiability is assignment-checking: given a 3CNF formula $\varphi$ and a setting of its variables to values $x_j \in \{0, 1\}$, does $\varphi(x)$ evaluate to $0$ or $1$? Clearly, this can be answered in a number of computation steps which is polynomial in the length of the description of $\varphi$. We thus say that this problem is in the class $\P$, consisting of all decision problems which can be solved in polynomial time.

The next relevant problem is satisfiability (3SAT): given a 3CNF formula $\varphi$, does there exist a setting $x$ of its inputs such that $\varphi(x) = 1$? Of course, if someone provides you with an $x$ such that $\varphi(x) = 1$ (a ``witness''), then the problem of verifying this is in $\P$. This means that 3SAT falls into the class NP of problems whose positive solutions are verifiable in polynomial time. In fact, by the famous Cook-Levin theorem~\cite{Cook,Levin}, 3SAT is also ``hard'' for the class\footnote{In fact, 3SAT is ``NP-complete,'' meaning that it is both NP-hard and also in NP.}\,\,~NP, i.e., any other problem $L$ in NP can be polynomial-time reduced to solving 3SAT. In other words, an algorithm for 3SAT can also be used to solve any other problem in NP, with at most a polynomial number of additional steps. But how do we directly attack 3SAT itself? The obvious approach is to simply try all possible assignments, of which there are exponentially many. If one can do significantly better is one of the biggest open questions in science: is P = NP? Of course, the conjectured state of affairs is P $\neq$ NP.

The third relevant problem counts satisfying assignments ($\#3$SAT): given a 3CNF formula $\varphi$, how many assignments $x$ satisfy $\varphi(x) = 1$? This problem is contained in (and is hard for) the class $\#$P. This class demands that the output to the algorithm is a number $m$ such that there exist exactly $m$ distinct proofs, each of which is polynomial-time checkable; in the case of $\#3$SAT, these proofs are the satisfying assignments themselves, and the polynomial-time checker is the first algorithm discussed above. Clearly, $\#3$SAT is at least as hard as 3SAT, so NP $\subseteq \#$P. It is a widely-believed conjecture that the containment is strict, i.e. that NP $\neq \#$P. Formally, this last statement is invalid, since NP contains only decision problems, while $\#$P contains counting problems. We will instead consider the conjecture $\#\P \not\subseteq \FP^\NP$. Here $\FP$ denotes the class of \emph{functions} which can be computed in polynomial time, rather than simply decision problems. $\FP^\NP$ is then the class of functions computable by a polynomial-time algorithm with access to an $\NP$ oracle\footnote{Informally, this oracle can be viewed as a ``black-box'' subroutine which the $\FP$-algorithm can invoke; the subroutine solves, in a single timestep, any problem lying in $\NP$ (e.g., $3$SAT).}. Strictly speaking, this is a weaker statement, since $\FP^\NP$ contains the class coNP, which is thought to be distinct from $\NP$.

\subsection{Quantum Circuits}

Recall that a boolean circuit consists of wires and gates. Each wire carries a bit, and each gate performs a local boolean operation. The circuit implements the boolean function defined by the composition of these operations. Boolean circuits are a useful abstraction for thinking about classical, digital computers. An analogous abstraction can also be defined for quantum-mechanical computation; this is a quantum circuit. Quantum circuits also consist of wires and gates, but each wire now carries a qubit state, i.e., a unit vector in the space $\CC_2$ (equipped with a preferred basis $\{\ket 0, \ket 1\}$ corresponding to the values of a classical bit.) Each gate is a unitary operator acting on a constant number of qubits, and leaving the remaining qubits fixed. The circuit implements the unitary operator defined by the composition of these local gates. Unlike with boolean circuits, the total number of qubits does not change as gates are applied.

While one can consider arbitrary local unitary gates, we will only make use of the following three:
$$
X = 
\left(\begin{matrix}
0 & 1 \\
1 & 0 
\end{matrix}\right)
\qquad \text{and} \qquad
H =  
\frac{1}{\sqrt{2}}
\left(\begin{matrix}
1 & 1 \\
1 & -1 
\end{matrix}\right)
\qquad \text{and} \qquad
T: \ket{x, y, z} \to \ket{x , y , z \oplus xy}\,.
$$
The last gate is the controlled-controlled-NOT gate (sometimes called the Toffoli gate), and is universal for classical computation. The set $\{H, T\}$ is a universal gate set for quantum computation~\cite{Aharonov1, Shi}. Although the gate $X$ might seem extraneous, note that applying it using only the set $\{H, T\}$ requires an additional resource: qubits which are initialized in the $\ket{1}$ state.

Starting with the formalism of quantum circuits, one may define a model of quantum computation~\cite{NielsenChuang}.  It is natural to ask if this model depends on the choice of gate set in some essential way. The Solovay-Kitaev theorem is a basic result that answers this question in the negative: all universal gate sets define the same model of efficient quantum computation. More precisely, any universal gate set can simulate any other, with arbitrary operator-norm precision and polylogarithmic overhead.~\footnote{We remark that $\{H, T\}$ is not universal in the sense of Solovay-Kitaev, as it does not densely generate $SU(d)$. Nonetheless, it gives rise to the same model of computation. This detail is not relevant to us, since we will only use Solovay-Kitaev in order to simulate the set $\{H, T\}$ with some other, properly dense set.}

\begin{thm} [Solovay-Kitaev~\cite{NielsenChuang}]\label{thm:SolovayKitaev}
Let $S$ be a finite set of unitary operators which is closed under inverses and spans a dense subset of $SU(d)$.  Then for any $U \in SU(d)$ and any $\varepsilon>0$, there exists a composition $U' = G_1 \circ G_2 \cdots \circ G_m$ of operators from $S$ such that $\|U' - U\|<\varepsilon$ and $m \in O(\log^4 (1/\varepsilon))$.
\end{thm}

\subsection{Calculating matrix entries of quantum circuits is $\#\P$-hard}
\vspace{3mm}

We first show that exponentially accurate additive approximation of a matrix entry of a quantum circuit is $\#\P$-hard. This fact appears to be folklore in the quantum computation community. In further discussions, it will be important to remember that a ``quantum circuit'' is a \emph{classical description} of a unitary operator, as a list of numbers describing the sequence of local unitary gates and which qubits they are to be applied to.

\vspace{3mm}

\begin{bb}\label{circuit-problem}
Given a quantum circuit $C$ on $n$ qubits over the gate set $\{H, X, T\}$, output a number $x$ such that $| x - \bra{0^n} C \ket{0^n} | \leq 2^{-n}.$
\end{bb}
\vspace{3mm}

\begin{thm}\label{thm:circuit-hard}
Problem 1 is $\#\P$-hard.
\end{thm}

\begin{proof}
We assume an oracle $\mathcal O$ which, on input a quantum circuit $C$, outputs a number $\mathcal O(C)$ satisfying the conditions of the theorem. We will show how to use $\mathcal O$ to solve $\#3$SAT in classical deterministic polynomial time. Since $\#3$SAT is $\#\P$-hard, the result will follow.
\vspace{2mm}

Let $\varphi:\{0,1\}^{n-1} \to \{0,1\}$ be a 3-CNF formula.  We first create a reversible $n$-bit circuit $C_\varphi$ for $\varphi$:
\begin{equation}\label{eq:reversible-formula-circuit}
C_\varphi:\ket{x}\ket{y} \mapsto \ket{x}\ket{y \oplus \varphi(x)}\,.
\end{equation}
This is done by expressing each of the operators $\{\neg,\land,\lor\}$ in the description of $\varphi$ by appropriate sequences of the Toffoli gate~\cite{NielsenChuang}. Since these sequences are each of constant length, the size of $C_\varphi$ is linear in the size of the description of $\varphi$. We remark that $C_\varphi$ actually needs more than just $n$ bits, in order to implement each clause. In total, $n + cn^3$ bits will suffice, where $c$ is some universal constant; the $cn^3$ additional bits must be initialized to $0$ and are not changed by $C_\varphi$. For the sake of notational simplicity, we suppress these additional bits and assume that $C_\varphi$ is precisely of the form \eqref{eq:reversible-formula-circuit}.
\vspace{2mm}

Now let $C_\varphi'$ be the following quantum circuit over $\{H, X, T\}$:
\vspace{5mm}
$$
\Qcircuit @C=1em @R=1.0em {
&\qw 	&\gate{H} &\multigate{3}{\qquad C_\varphi \qquad} 	&\gate{H} &\qw 	&\qw \\
&&\vdots&&\vdots&&&&\\
&\qw 	&\gate{H} &\ghost{\qquad C_\varphi \qquad} 		&\gate{H} &\qw 	&\qw \\
&\gate{X} &\gate{H} &\ghost{\qquad C_\varphi \qquad}		&\gate{H} &\gate{X} 	&\qw
}
$$
\vspace{5mm}

\noindent Define $\ket{-} := H X \ket{0} = \ket{0} - \ket{1}$ and note that
$$
C_\varphi : \ket{x} \ket{-} \mapsto (-1)^{\varphi(x)} \ket{x} \ket{-}\,.
$$
One also easily checks that
$$
H^{\otimes n} \ket{0^n} = \frac{1}{2^{n/2}} \sum_{x \in \{0, 1\}^n} \ket{x}\,.
$$
Using the above, we may compute the relevant matrix entry of the circuit $C_\varphi'$, as follows.
\begin{align*}
\bra{0^n} C_\varphi' \ket{0^n} 
&= \frac{1}{2^{n-1}} \sum_{y, x \in \{0, 1\}^{n-1}} \bra{y} \bra{-} C_\varphi \ket{x} \ket{-} \\
&= \frac{1}{2^{n-1}} \sum_{x \in \{0, 1\}^{n-1}} (-1)^{\varphi(x)} \sum_{y \in \{0, 1\}^{n-1}} \innerprod{y}{x} \innerprod{-}{-} \\
&= \frac{1}{2^{n-1}}\sum_{x \in \{0,1\}^{n-1}} (-1)^{\varphi(x)}\\
&= 1 - \frac{\#\varphi}{2^{n-2}}\,,
\end{align*}
where $\#\varphi = |\{ x \in \{0, 1\}^{n-1} : \varphi(x) = 1\}|$ is the number of satisfying assignments of $\varphi$. When we apply the oracle $\mathcal O$ to the circuit $C_\varphi'$, we will receive a number $x$ which is at most $2^{-n}$ away from the above. It follows that
$$
\left|\#\varphi-2^{n-2}(1-x)\right| \leq \frac{1}{4}. 
$$
We can thus calculate $\#\varphi$ exactly by finding the closest integer to $2^{n-2}(1-x)$.
\end{proof}

\section{Main results}\label{sec:results}

\subsection{Calculating the WRT invariant is $\#\P$-hard}

Our first result shows that computing a certain additive approximation of the WRT invariant is a $\#\P$-hard problem. This is essentially the exponential-accuracy version of the main theorem in~\cite{BQP2}, which states that a weaker approximation level is BQP-hard. Our proof provides many of the details left out of~\cite{BQP2}. The basic ingredients of the proof are the Solovay-Kitaev theorem (\expref{Theorem}{thm:SolovayKitaev}), the density theorem of Larsen and Wang (\expref{Theorem}{thm:density}), and \#\P-hardness of quantum circuit entries (\expref{Theorem}{thm:circuit-hard}).

\begin{bb}\label{wrt-problem}
Given a word $\alpha$ in the standard generators of MCG($\Sigma_g$), output a number $x$ such that
$$
|x - \WRT_r(M_{g, \alpha})| \leq \frac{\D^{g-1}}{2^{g+1}}.
$$
\end{bb}

\begin{thm}\label{thm:wrt-hard}
For prime $r \geq 5$, \expref{Problem}{wrt-problem} is $\#\P$-hard.
\end{thm}

\begin{proof}
We will assume an oracle $\mathcal O$ for \expref{Problem}{wrt-problem}, and show how to use it (in classical, deterministic polynomial time) to solve \expref{Problem}{circuit-problem}. The result will then follow from \expref{Theorem}{thm:circuit-hard}. 

The starting point is an $n$-qubit quantum circuit $C$, and the goal is to efficiently produce a description of a Heegaard splitting $(g, \alpha)$ whose WRT invariant is exponentially close to $\bra{0^n}C\ket{0^n}$. The genus $g$ will be equal to the number of qubits $n$. Starting from the description of $C$, we will find a word $\alpha$ in the standard generators of $\MCG(\Sigma_g)$ such that $\rho_{r, g}(\alpha)$ closely approximates $C$ on an appropriate subspace of $\mathcal H_{r, g}$. For the comparison to make sense, this subspace should be isomorphic to $\CC_2^{\otimes n}$, i.e., the space of $n$ qubits. We choose the subspace spanned by labelings of the standard spine which take the form shown in \expref{Figure}{encoding}. The bits of an $n$-qubit basis state $\ket{z_1, \dots, z_g}$ label the $g = n$ loops, with the rest of the edges labeled zero. We remark that since $(k, k, 0)$ is a valid (degenerate) triangle for any $k$, the fusion rules allow the $z_j$ to carry arbitrary labels, as desired. The result is an embedding
$\phi : \CC_2^{\otimes n} \hookrightarrow \mathcal H_{r, g}$ of Hilbert spaces, with a corresponding homomorphism $\Phi : SU(\CC_2^{\otimes n}) \hookrightarrow SU(\mathcal H_{r, g})$ of unitary groups.
\begin{figure}[h]
\begin{center}
\includegraphics[width = 0.95\textwidth]{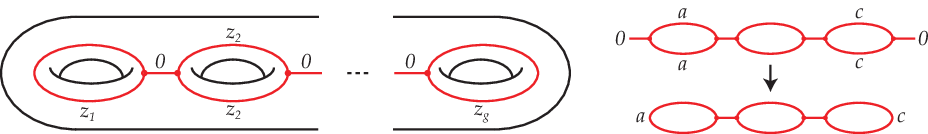}
\fcaption{\label{encoding} Encoding a $g$-qubit basis state $\ket{z_1, \dots, z_g}$ into a basis vector of $\H_{r,g}$ (left),  and the isomorphism between $\mathcal H_{r, 3}^{(0, 0)}$ and $\mathcal H_{r, 3}$ (right).}
\end{center}
\end{figure}

The circuit $C$ is completely described by some sequence of unitary gates. We will simulate each gate's action on the encoded subspace by using a sequence of Dehn twists. We remark that this cannot simply be done by applying the density theorem and Solovay-Kitaev to all of $\rho_{r, g}$. Indeed, the Solovay-Kitaev algorithm necessarily scales at least exponentially with the number of qubits~\cite{NielsenDawson}. The Dehn words produced in this way would thus have exponential length. For this reason, we will instead decompose the problem in such a way that we need only apply Solovay-Kitaev on spaces of constant dimension (relative to $n$.)

To start, let us take the first gate. It acts on a contiguous segment of at most three qubits; this segment of qubits corresponds to a segment of the $g$-torus containing three handles, say handles $1 < k, k+1, k+2 < g$. If we cut out this segment as in \expref{Figure}{torus-cutout}, we get a corresponding decomposition of Hilbert spaces:
\begin{equation}\label{eq:cut-space-decomposition}
\mathcal H_{r, g} = \bigoplus_{(a, b) \in L} \mathcal H_{r, k-1}^{(a)} \otimes \mathcal H_{r, 3}^{(a, b)} \otimes \mathcal H_{r, g-k-2}^{(b)}\,.
\end{equation}
\begin{figure}[h]
\begin{center}
\includegraphics[width = 0.95\textwidth]{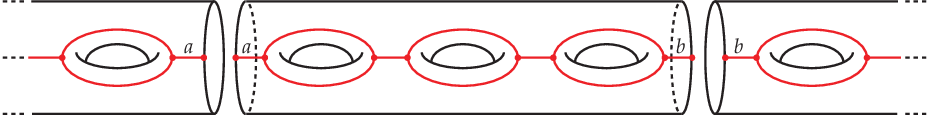}
\fcaption{\label{torus-cutout} Each 3-qubit gate in the circuit corresponds to some three-handled segment of $\Sigma_g$.}
\end{center}
\end{figure}
Due to our encoding, it suffices to consider only the subspace corresponding to $a = b = 0$. The first gate, under the encoding, is some particular element $\Phi(U_1) \in SU(\mathcal H_{r, 3}^{(0, 0)})$. Observe that $\mathcal H_{r, 3}^{(0, 0)}$ is isomorphic to $\mathcal H_{r, 3}$, as shown in \expref{Figure}{encoding}. Any basis vector of the former has two ``free'' edges, one on the left and one on the right, with both carrying the zero label. Consider one of these (say the left.) Due to the fusion rules, the two adjacent edges (which together form a loop around the leftmost handle) must carry the same label. The isomorphism removes the free edge, and then joins these two edges into a single loop carrying that label; the same operation is performed on the right side. Due to this isomorphism, we can assume that in fact $\Phi(U_1) \in SU(\mathcal H_{r, 3})$. By the density theorem (\expref{Theorem}{thm:density}) and Solovay-Kitaev theorem (\expref{Theorem}{thm:SolovayKitaev}), there is a word $\alpha_1 \in \MCG(\Sigma_3)$ of length polynomial in $n$, such that
$\| \rho_{r, 3}(\alpha_1) - \Phi(U_1) \| \leq |C|^{-1}2^{-n-1}$. By the decomposition \eqref{eq:cut-space-decomposition} and the embedding $\MCG(\Sigma_3) \hookrightarrow \MCG(\Sigma_g)$ which sends the three handles to positions $k, k+1, k+2$, we can instead write
$$
\left\| \rho_{r, g}(\alpha_1) - \Phi(C_1) \right\| \leq |C|^{-1}2^{-n-1}\,,
$$ 
where $C_1 = \one_{k-1} \otimes U_1 \otimes \one_{g-k-2}$. We will adopt this notation for the remaining operators $C_j$ comprising $C$, so that we can simply write $C = C_{|C|} \cdots C_1$, as operators in $SU(\CC_2^{\otimes n})$.

We now apply the above process to each $C_j$, and compose the resulting words $\alpha_j$ to get a word $\alpha = \alpha_{|C|} \alpha_{|C|-1} \cdots \alpha_2 \alpha_1$ in the generators of $\MCG(\Sigma_g)$. We then have
$$
\left\| \rho_{r, g}(\alpha) - \Phi(C) \right \| \leq \sum_{j=1}^{|C|} \left\| \rho_{r, g}(\alpha_j) - \Phi(C_j) \right \| \leq \frac{1}{2^{n+1}}\,.
$$
Note that we have maintained zero labels on all of the edges of the spine not involved in the encoding. From that and the above, it follows that
\begin{equation}\label{eq:circuit-approx}
\left| \bra{0} \rho_{r, g}(\alpha)\ket{0} - \bra{0}C\ket{0}\right| 
= \left| \bra{0} \rho_{r, g}(\alpha)\ket{0} - \bra{0} \Phi(C) \ket{0} \right| 
= \left| \bra{0} \rho_{r, g}(\alpha) - \Phi(C) \ket{0}\right| 
\leq \frac{1}{2^{n+1}}\,.
\end{equation}

We are now ready to solve \expref{Problem}{circuit-problem} for the circuit $C$. We first input the Heegaard splitting $(g, \alpha)$ into the WRT oracle $\mathcal O$, and receive a number $x$; we then output $\mathcal D^{1-g}x$. To show that this solves \expref{Problem}{circuit-problem}, note that
\begin{equation}\label{eq:final-approx}
\left| \mathcal D^{1-g} x - \bra{0} \rho_{r, g}(\alpha) \ket{0} \right| \leq \frac{1}{2^{n+1}}
\end{equation}
by the definition of \expref{Problem}{wrt-problem}. By the triangle inequality applied to the above and \eqref{eq:circuit-approx}, we get
$$
\left| \mathcal D^{1-g} x - \bra{0} C \ket{0} \right| \leq \frac{1}{2^{n}}
$$
as desired.

\end{proof}

For the main theorem of the paper, we will actually only need the fact that exact calculation of the WRT invariant is $\#\P$-hard. This is an immediate corollary of the above, stated as \expref{Theorem}{thm:intro1} above.

\begin{bb}\label{problem:exact-wrt}
Given a sequence of generators $w = (w_1, \dots, w_k)$ of MCG($\Sigma_g$), output $\WRT(M_{g, w})$.
\end{bb}

\begin{thm-hand}[1.1]\label{cor:wrt-hard}
Problem \ref{problem:exact-wrt} is $\#\P$-hard.
\end{thm-hand}

\subsection{Implications for Heegaard splittings and the $r$-distance}

Our main result applies the \#\P-hardness of the WRT invariant to prove the existence of a family of $3$-manifold diagrams which is of interest in understanding the $r$-distance. This result is an analogue of Freedman's main theorem~\cite{Freedman1}, in the setting of Heegaard splittings and $3$-manifolds.

\begin{thm-hand}[1.3]\label{thm:main}
Assume $\#\P \not\subseteq \FP^{\NP}$, and choose prime $r \geq 5$. For any $\RR^+$-valued polynomial $p$ in one variable, there exists an infinite family of Heegaard splittings $(g_j, \alpha_j)$ with the following property: for any family $(h_j, \beta_j)$ satisfying $\dist_r((g_j, \alpha_j), (h_j, \beta_j)) \leq O(p(g_j))$, it is the case that $h_j \in \Omega(\log(g_j))$.
\end{thm-hand}

\begin{proof}
We begin by assuming that the theorem conclusion is false. It then follows that that there exists a polynomial $p$ and constants $c_1, c_2$ such that the following holds for all but finitely many Heegaard splittings: for any $(g, \alpha)$ there exists $(g', \alpha')$ such that 
\begin{equation}\label{eq:contradiction}
\dist_r((g, \alpha), (g', \alpha')) \leq c_1p(g)
\qquad \text{and} \qquad
g' < c_2 \log(g)\,.
\end{equation}
We will show that this implies an $\FP^{\NP}$-algorithm $\mathcal A$ for \expref{Problem}{problem:exact-wrt}, i.e., for exactly calculating the WRT invariant. Since $\FP^{\NP}$ is closed under polynomial-time reductions, and exact WRT is $\#\P$-hard (\expref{Theorem}{thm:wrt-hard}), it will follow that $\#\P \subseteq \FP^{\NP}$, a contradiction. 

We remark that, since we are only interested in the $\#\P$-hard case, we may assume that $|\alpha|$ (and by \eqref{eq:contradiction} also $|\alpha'|$) is polynomial in $g$. Indeed, in the proof of \expref{Theorem}{thm:wrt-hard}, the genus corresponds to the number of variables in a 3CNF formula, and the word $\alpha$ is produced from a description of the formula itself, which is at most polynomial in the number of variables.

Since it is an $\FP^{\NP}$-algorithm, $\mathcal A$ can make use of a polynomial-size ``witness'' in order to calculate WRT exactly, in classical deterministic polynomial time. For calculating $\WRT(g, \alpha)$, the witness will be a description of the sequence of diagram moves (handle-slides, stabilizations, and $r$-moves) necessary to transform $(g, \alpha)$ to the associated splitting $(g', \alpha')$ from \eqref{eq:contradiction}. By the definition of $r$-distance, this sequence admits a description of length polynomial in $g$. It thus suffices to show that $\mathcal A$ can verify the correctness of each move in the sequence in polynomial time. This is fairly straightforward:
\begin{enumerate}
\item a handle-slide is described by a pair $(\gamma, \gamma')$ of words in the generators; to check validity, it suffices to check that both $\gamma$ and $\gamma'$ describe elements of the handlebody subgroup; this can be done in time polynomial in $|\gamma| + |\gamma'|$~\cite[Theorem 6.4]{Sch};
\item stabilization and destabilization are described by a single bit (e.g., $0$ to stabilize, $1$ to destabilize); to check validity, we need only check that the word has the appropriate form, i.e., \eqref{eq:stabilize} or \eqref{eq:destabilize};
\item the $r$-move is always valid, and is described by a generator index $j$ and another index indicating where $\sigma_j^{2r}$ should be inserted; an inverse $r$-move simply indicates the position of a subword of the form $\sigma_j^{2r}$, whose existence is trivial to check.
\end{enumerate}
The final verification step is simply a string comparison between $(g', \alpha')$ and the splitting produced by applying the witness sequence to $(g, \alpha)$. By \expref{Theorem}{thm:r-distance}, this also verifies that $\WRT_r(g, \alpha) = \WRT_r(g', \alpha')$.

The final step is to show how $\mathcal A$ can calculate $\WRT_r(g', \alpha')$ in time polynomial in $g$. By the second property in \eqref{eq:contradiction}, the dimension of the relevant Hilbert space satisfies
$$
\mathcal H_{g', r} < r^{3g'} \leq r^{3 c_2 \log(g)} \in O(\text{poly}(g))\,,
$$
where the first inequality comes from counting spine edges. All of the pieces of data associated to $r$ and described in \expref{Section}{sec:wrt} can be written down explicitly and exactly in time and space which does not depend on the genus or the word length. For each standard generator $\sigma_j$ of $\MCG(\Sigma_g)$, the matrix $\rho_{g', r}(\sigma_j)$ is some product of a diagonal matrix with a constant number of $S$-matrices and $F$-matrices. It follows that the entire matrix for each generator can be computed explicitly and exactly in time polynomial in the dimension of $\mathcal H_{g', r}$. We can thus compute the matrix $\rho_{g', r}(\alpha')$ by computing and then multiplying together the relevant $|\alpha'|$-many matrices, which takes time $O(\text{poly}(g, |\alpha'|)) = O(\text{poly}(g))$. We then simply output the (scaled) matrix entry corresponding to the WRT invariant. This completes the description of the algorithm $\mathcal A$, and our proof.
\end{proof}

\section{Extensions}\label{sec:tv}

\subsection{Strengthening the hardness result}\label{sec:WRTapprox}

We now show that any value-distinguishing approximation of the WRT invariant is still $\#\P$-hard. Our proof is an adaptation of Kuperberg's proof of the same fact for the Jones polynomial~\cite{Kuperberg}. Since the basic steps of our proof are essentially the same as that of Kuperberg, we omit some details. The basic ingredients still include the Solovay-Kitaev theorem (\expref{Theorem}{thm:SolovayKitaev}) and the Larsen and Wang density result (\expref{Theorem}{thm:density}); in addition, we will now also need a strengthening of \expref{Theorem}{thm:circuit-hard}, namely Aaronson's result that PostBQP = PP~\cite{AaronsonPostBQP}.

Before we begin, we need a few new pieces of notation. We only give brief descriptions here, since a thorough explanation would go well beyond the scope of the paper.
\begin{itemize}
\item Let $\{0,1\}^*$ denote the set of bitstrings of arbitrary length. A function $f : \{0,1\}^* \rightarrow \RR^+$ is said to have a \emph{multiplicative approximation} if there exists an algorithm $\mathcal A$ which, on input $x$, outputs a number $\mathcal A(x)$ that is within some bounded factor of $f(x)$, in time polynomial in $|x|$.
\item A function $f : \{0,1\}^* \rightarrow \RR^+$ is said to have a FPTEAS (fully polynomial-time exponential approximation scheme) if there exists an algorithm $\mathcal A$ which, on input $x$ and $\epsilon > 0$, outputs a number $\mathcal A(x)$ that is within a $1 + \epsilon$ factor of $f(x)$, in time polynomial in $|x|$ and $\log ( 1 / \epsilon )$.
\item The class of decision problems solvable by quantum algorithms in polynomial time with bounded error is denoted BQP. We can significantly strengthen this class by adding the power to \emph{post-select} on events which have inverse-exponential probability. Take the following example ``post-selected algorithm'': we execute a quantum circuit, and measure all the qubits; conditioned on the first qubit measuring to $0$ (regardless of its probability), we output the second qubit. The class of problems solvable by such algorithms is called PostBQP. 
\item PP is the class of decision problems solvable by classical polynomial-time algorithms with error probability at most $1/2$. It is a significant strengthening of the class BPP of efficient polynomial-time probabilistic algorithms, which can err with probability strictly bounded away from $1/2$ (which then permits efficient amplification by repetition and majority-voting.) An important result of Aaronson states that PostBQP = PP~\cite{AaronsonPostBQP}. PP is also closely related to $\#\P$: for example, PP-hardness implies $\#\P$-hardness~\cite{Kuperberg}.
\end{itemize}
We are now ready to prove the strengthening of our hardness result.

\begin{thm-hand}[1.2]\label{thm:main2}
Let $r \geq 5$ be prime, and $0 < a < b$ real. Given an integer $g$, a word $\alpha$, and a promise that $|\WRT_r(M_{g, \alpha})| < a$ or $|\WRT_r(M_{g, \alpha})| > b$, it is $\#\P$-hard to decide which is the case.
\end{thm-hand}
\begin{proof}
We first show that multiplicative approximation of $|\WRT_r(M_{g, \alpha})|$ is $\#\P$-hard. The starting point is to observe that the proof of \expref{Theorem}{thm:wrt-hard} can be adapted to show the following fact. Given a quantum decision algorithm $\mathcal A$ and an input $x$, there exists a Heegaard splitting $(g, \alpha)$ such that the acceptance probability $r(x)$ of $\mathcal A$ satisfies
\begin{equation}\label{eq:FPTEAS}
r(x) \approx \frac{|\WRT_r(M_{g, \alpha})|^2 }{|\mathcal D|^{2g}}\,,
\end{equation}
where $\approx$ means well-approximated in the sense of FPTEAS, i.e., to within a factor $1 + \epsilon$ in time polynomial in $|x|$ and $- \log(\epsilon)$. To prove this, one first shows that the circuit of $\mathcal A$ (which is to be applied to $x$) can be transformed into another circuit $C$ such that the probability of acceptance is encoded in $|\bra{0^n} C \ket{0^n}|^2$ (Proposition 2.3 of~\cite{Kuperberg}.) One then follows the steps of \expref{Theorem}{thm:wrt-hard} (but now with the more general circuit $C$, rather than the SAT-motivated circuit) to find the suitable $g$ and $\alpha$. By Proposition 2.14 in~\cite{Kuperberg}, multiplicative approximation of the acceptance probability of quantum algorithms (i.e., the quantity $r(x)$) is PostBQP-hard; by \eqref{eq:FPTEAS}, multiplicative approximation of $|\WRT_r(M_{g, \alpha})|$ is also PostBQP-hard. Aaronson's theorem~\cite{AaronsonPostBQP} (PostBQP = PP) and Proposition 2.1 in ~\cite{Kuperberg} (PP-hard implies $\#\P$-hard) complete the argument.

Next, we want to extend the proof to handle any value-distinguishing approximation. Choose $0 < a < b$ real. By Lemma 2.12 in ~\cite{Kuperberg}, PostBQP-hardness of the right-hand side of \eqref{eq:FPTEAS} implies that there exists a polynomial $p$ such that for any constant $c > 1$ of our choosing, it is $\#\P$-hard to decide which of these two inequalities hold:
\begin{equation}\label{eq:WRT-approx-2}
\frac{|\WRT_r(M_{g, \alpha})|^2 }{|\mathcal D|^{2g}} 
\qquad
\begin{cases} 
&< 2^{-p(|x|)}\\
&> c 2^{-p(|x|)} \,.
\end{cases}
\end{equation}
We will manipulate the splitting $(g, \alpha)$ to produce another splitting $(g', \alpha')$, such that the above fact will imply $\#\P$-hardness of deciding the gap $a < b$ for $|\WRT_r(M_{g', \alpha'})|$. Note that, almost by definition, $\WRT_r$ is multiplicative under taking connected sums (for a proof, see Section 5 in~\cite{KirbyMelvin91}.) Our manipulations will always consist of attaching many copies of a constant-size Heegaard splitting via connected sum. This will allow us to either increase or decrease the absolute value of the invariant in order to turn the interval of approximation in \eqref{eq:WRT-approx-2} into the interval we want, i.e., $(a, b)$. Recall that $|\mathcal D| > 1$ and that $|x|$ and $g$ are polynomially-related, since each bit of $x$ requires a constant number of handles in the encoding from the proof of \expref{Theorem}{thm:wrt-hard}. 

First, if $|\mathcal D|^{2g} \ll 2^{p(|x|)}$ in the limit, then $|\mathcal D|^{2g} 2^{-p(|x|)} \rightarrow 0$. We thus make $M_{g', \alpha'}$ by attaching $m = \text{poly}(|x|)$ copies of the genus-one Heegaard splitting of the three-sphere to the initial manifold $M_{g, \alpha}$. We choose $m$ sufficiently large that $|\mathcal D|^{2g + 2m} 2^{-p(|x|)}$ is asymptotically bounded by a constant. Applying \eqref{eq:WRT-approx-2}, we now have the decision problem
$$
|\WRT_r(M_{g', \alpha'})|^2 =
|\WRT_r(M_{g, \alpha})|^2 |\WRT_r(S^3)|^{2m}
\qquad
\begin{cases} 
&< |\mathcal D|^{2g + 2m} 2^{-p(|x|)}\\
&> c |\mathcal D|^{2g + 2m} 2^{-p(|x|)}
\end{cases}
$$
for any constant $c$ of our choosing. We can now clearly select $c$ so that deciding the $(a, b)$ gap for $M_{g', \alpha'}$ allows us to decide \eqref{eq:WRT-approx-2}.

On the other hand, if $|\mathcal D|^{2g} \gg 2^{p(|x|)}$ in the limit, then we attach $m$ copies of a two-genus splitting $(2, \beta)$ where $|\WRT_r(2, \beta)|$ is a small constant, say $|\mathcal D|^{-1}$. Now $m$ is chosen sufficiently large that 
$$
|\mathcal D|^{2g} 2^{-p(|x|)} |\WRT_r(2, \beta)|^{2m}
$$ 
is asymptotically bounded, and we pick $c$ as before. To find a suitable $\beta$, we can pick any unitary $U \in SU(\mathcal H_{r, 2})$ whose matrix entry $|\bra{0}U\ket{0}|$ is sufficiently small, and then apply the Solovay-Kitaev theorem to find the corresponding $\beta$.
\end{proof}

\subsection{Turaev-Viro invariants and triangulations}

We now sketch out how to adapt our results to the case of Turaev-Viro invariants and triangulated 3-manifolds. Recall that the Turaev-Viro (or TV) invariant is also parameterized by an integer $r$, and can be defined using the same data described in \expref{Section}{sec:wrt}. The value of TV is given by the contraction of a certain tensor network constructed by placing $F$-tensors on each tetrahedron in the triangulation. The six indices of each copy of the $F$-tensor are associated to the six edges of the corresponding tetrahedron. The value is thus a sum over all labelings of the edge set of the triangulation by elements of $L$. More precisely, for a triangulation $T$ with vertex set $V$, we set
\begin{equation} \label{eq:turaev-viro}
\hspace{-0.5ex} \TV_r(T) = 
\mathcal D^{-2 |V|}\hspace{-0.2ex} \sum_{\textrm{labelings}} \hspace{0.4ex} \prod_{\textrm{edges}} d_e \prod_{\textrm{tetrahedra}} \frac{F^{ijm}_{kln}}{\sqrt{d_m d_n}}
\enspace 
\end{equation}
It is well-known that the Turaev-Viro invariant of a manifold is equal to the squared-modulus of the WRT invariant~\cite{BQP1, TV}.

We now sketch out how to adapt the proof of $\#\P$-hardness from WRT to TV. This can be strengthened to show $\#\P$-hardness of approximation up to a constant, just as was done for WRT in \expref{Section}{sec:WRTapprox}.

\begin{thm}\label{tv-problem}
Let $r \geq 5$ be prime, and $0 < a < b$ real. Given a triangulation $T$, and a promise that $|\TV_r(T)| < a$ or $|\TV_r(T)| > b$, it is $\#\P$-hard to decide which is the case.
\end{thm}
\begin{proof}
(Sketch.) In \expref{Theorem}{thm:wrt-hard}, we showed how to map a 3CNF formula $\varphi$ to a Heegaard splitting, such that the value of WRT closely approximates a positive, real quantity from which we can easily surmise $\#\varphi$. It's not hard to see that having the square of this quantity also suffices for computing $\#\varphi$. The proof of \expref{Theorem}{thm:intro3} carries over as well, the only change being that we are now considering the squared quantity.

The remaining step is thus to show how to efficiently triangulate a Heegaard splitting $(g, \alpha)$. This can be done by first triangulating the two handlebodies, in such a way that Dehn twists along all the canonical curves can be applied via flip moves (i.e., Pachner 2-2 moves~\cite{Pachner}). Such a handlebody triangulation can be created using only a constant number of tetrahedra per handle. To attach the handlebodies, we need a triangulation of the mapping cylinder of $\alpha \in \MCG(\Sigma_g)$. The cylinder can be constructed in layers, with one layer for each generator appearing in the word $\alpha$. Each layer then corresponds to a single Dehn twist, which can be performed by a sequence of flip moves. Each flip move corresponds to a single tetrahedron in the layer (recall that a flip move can be implemented by gluing one tetrahedron onto the two relevant triangles.) Further details on triangulating Heegaard splittings are given in~\cite{BQP3}.
\end{proof}

To adapt \expref{Theorem 1.3}{thm:main} to triangulations, we will need a notion of $r$-distance and a notion of girth. The $r$-distance is defined by local moves. The homeomorphism-preserving moves are the two 3-dimensional Pachner moves, the 3-2 move and the 1-4 move~\cite{Pachner}. The $r$-move is defined as before, by allowing the insertion of (triangulated) Dehn twists of order $2r$. As before, the details of the moves themselves are not so important, so long as they all preserve the Turaev-Viro invariant and can be verified efficiently. In the case of our moves, this follows from Pachner's result and the fact that $\TV = |\WRT|^2$.

The girth of a triangulation is defined as the maximum width of the thinnest ordering, i.e.,
$$
g(T) = \min_{\textrm{orderings of $T$}} \hspace{0.4ex} \max_{1 \leq k \leq |T|} \left| \partial \left( \cup_{j=1}^k T_j\right) \right|\,.
$$
where ``ordering of $T$'' means an ordering of the set of tetrahedra, and $|\partial T'|$ denotes the number of triangles on the boundary of $T'$. With these definitions in place, we are ready to state our final theorem.

\begin{thm}
Assume $\#\P \not\subseteq \FP^{\NP}$, and choose prime $r \geq 5$. For any $\RR^+$-valued polynomial $p$ in one variable, there exists an infinite family of $3$-manifold triangulations $T_j$ with the following property: for any triangulation family $G_j$ satisfying $\dist_r(T_j, G_j) \leq O(p(|T_j|))$, it is the case that $g(G_j) \in \Omega(\log(|T_j|))$.
\end{thm}
\begin{proof}(Sketch.)
The proof is essentially the same as before. For a contradiction, we assume the conclusion of the theorem to be false, and conclude that there is an FP$^{\NP}$-algorithm for the Turaev-Viro invariant. The polynomial-size witness now consists of two parts: first, the sequence of Pachner and $r$-moves that map $T_j$ to $G_j$, and second, an ordering of $G_j$ which has the promised girth $O(\log(|T_j|))$. The first part of the witness enables us to verify that $\TV_r(T_j) = \TV_r(G_j)$ in polynomial time. The second part ensures that $\TV_r(G_j)$ can be computed exactly in polynomial time, by contracting the tensor network in the specified order. At any given step of the contraction, the number of free indices is proportional to the number of triangles on the surface, which is logarithmic in $|T_j|$. The total dimension of the space is thus polynomial throughout.
\end{proof}

\vspace{2mm}
\section*{Acknowledgements}
\noindent
CL was supported by Caltech's Summer Undergraduate Research Thomas Lauritsen Fellowship, as well as by John Preskill with NSA/ARO grant W911NF-09-1-0442 and the Institute for Quantum Information and Matter (IQIM), an NSF Physics Frontiers Center with support from the Gordon and Betty Moore Foundation. Much of this work was completed while GA was a postdoctoral scholar at IQIM. GA acknowledges financial support from the European Research Council (ERC Grant Agreement no 337603), the Danish Council for Independent Research (Sapere Aude) and VILLUM FONDEN via the QMATH Centre of Excellence (Grant No. 10059). We thank the authors of \cite{BQP3, BQP1} for some of the diagrams used in this paper. We are thankful to Greg Kuperberg for comments on an earlier version.

$ $

\bibliographystyle{plain}

\bibliography{references}

\end{document}